\documentclass[12pt]{article}
\usepackage{amsmath}
\usepackage{multicol,caption}
\usepackage[dvips]{graphicx}
\usepackage{times}

\usepackage{amsthm}
\usepackage{amssymb}
\usepackage{graphicx}
\usepackage{amsfonts}
\usepackage{color}

\newtheorem{theorem}{Theorem}
\newenvironment{Figure}
  {\par\medskip\noindent\minipage{\linewidth}}
  {\endminipage\par\medskip}

\def\be{\begin{equation}}
\def\ee{\end{equation}}
\def\nn{\nonumber}
\def\l{\left}
\def\r{\right}
\def\ba{\begin{array}}
\def\ea{\end{array}}
\def\bea{\begin{eqnarray}}
\def\eea{\end{eqnarray}}
\def\bn{\begin{eqnarray*}}
\def\en{\end{eqnarray*}}

\def\p{\partial}
\def\f{\frac}

\def\bra{\langle}
\def\ket{\rangle}
\def\bfg{\begin{figure}}
\def\efg{\end{figure}}

\def\R{\mathbb{R}}
\def\C{\mathbb{C}}
\def\tr{\text{tr}}

\begin{document}
\title{Nonadiabatic holonomic one-qubit gates}
\author{Takumi Nitanda, Utkan G\"ung\"ord\"u, Mikio Nakahara\\
Kinki University\\
3-4-1 Kowakae, Higashi-Osaka, 577-8502 Japan}

\maketitle

\begin{abstract}
Adiabatic quantum gate implementation generally takes longer time,
which is disadvantageous in view of decoherence. In this report
we implement several essential one-qubit quantum gates nonadiabatically
by making use of a dynamical invariant associated with a Hamiltonian.
Moreover we require that these gates be holonomic, that is, the dynamical
phases associated with the gates vanish. Our implementation is
based on our recent work [J. Phys. Soc. Jpn. {\bf 83}, 034001 (2014)] and the gate parameters required
for the implementations are found by numerical optimization.
\end{abstract}

{\bf keywords:} nonadiabatic control, Aharonov-Anandan phase, holonomic
quantum gate.

\begin{multicols}{2}

\section{Introduction}
\subsection{Dynamical Invariants}

Let $H = H(t)\in M_n(\C)$ be a time-dependent Hamiltonian. A time-dependent Hermitian operator
$I = I(t) \in M_n(\C)$ is called a dynamical invariant 
(also known as the Lewis-Riesenfeld invariant \cite{lr}) if it satisfies
\be
i \frac{\partial I}{\partial t} = [H, I].
\ee
Let $|\psi(t) \ket$ be a solution of the time-dependent Schr\"odinger
equation 
\be\label{eq:sch}
i \f{d}{dt} |\psi(t) \ket = H |\psi(t)\ket.
\ee
Then $\bra \psi(t)|I |\psi(t)\ket$ is independent of time. 
In fact, 
\bn
\lefteqn{\f{d}{dt} \bra \psi(t)|I |\psi(t)\ket}\\
&=& \l(\f{d}{dt} \bra \psi(t)|\r)I |\psi(t)\ket +\bra \psi(t)|\dot{I}|\psi(t) \ket \\
& & + \bra \psi(t) |I \l(\f{d}{dt}|\psi(t) \ket\r)\\
&=& \bra \psi(t)|\l[i H I-i [H I-I H]- iIH\r]
|\psi(t) \ket\\
& =& 0.
\en

Let $\{\lambda_k\}$ be the set of eigenvalues of $I$ and 
$\{|\phi_k(t) \ket\}$ be the corresponding set of normalized
eigenvectors;
$I|\phi_k(t) \ket = \lambda_k |\phi_k(t) \ket$. It might seem
that $\lambda_k$ depends on time since $I$ does.
Observe, however, that
\bn
\dot\lambda_k &=& \f{d}{dt} \bra \phi_k(t)|I |\phi_k(t) \ket\\
&=& \l( \f{d}{dt} \bra \phi_k(t)|\r) I |\phi_k(t) \ket \\
& & + \bra \phi_k(t)|\dot{I}(t) |\phi_k(t) \ket\\
& & + \bra \phi_k(t)|I \l( \f{d}{dt} |\phi_k(t) \ket \r)\\
&=& \lambda_k \l( \f{d}{dt}\bra \phi_k(t)|\r) |\phi_k(t) \ket \\
& & -i \lambda_k \bra \phi_k(t)|(H-H)|\phi_k(t) \ket \\
& & + \lambda_k \bra \phi_k(t)| \l( \f{d}{dt}|\phi_k(t) \ket\r)\\
&=& \lambda_k \f{d}{dt} \l( \bra \phi_k(t)|\phi_k(t) \ket \r) = 0,
\en
where use has been made of the normalization condition 
$ \bra \phi_k(t)|\phi_k(t) \ket=1$.

The dynamical invariant
has the following spectral decomposition
\be
I = \sum_k \lambda_k|\phi_k(t)\ket \bra \phi_k(t)|,\quad \lambda_k \in \R.
\ee

\subsection{Solutions of the Schr\"odinger Equation}

Take $|\phi_k(0) \ket$ and consider a solution $|\psi_k(t)\ket$
of the Schr\"odinger equation $i\partial_t|\psi_k(t) \ket
= H |\psi_k(t) \ket$ such that
$|\psi_k(0)\ket = |\phi_k(0) \ket$. The solution $|\psi_k (t) \ket$
should not be confused with the $k$-th eigenvector of $H$. The
index $k$ simply states that the vector was initially the
eigenvector $|\phi_k(0) \ket$ of $I(0)$. 
\\
\\
\begin{theorem}
The solution $|\psi_k(t) \ket$ of the Schr\"odinger
equation (2) is given by 
\be
|\psi_k(t) \ket = e^{i \alpha_k(t)}|\phi_k(t) \ket
\ee
with
\be
\alpha_k(t) = \int_0^t \bra \phi_k(s)|[i \partial_s - H(s) ]|\phi_k(t) \ket ds.
\ee
\end{theorem}

\begin{proof}
It follows from 
$I |\phi_k \ket = \lambda_k |\phi_k \ket$ that
\bn 
\lambda_k |\dot{\phi}_k \ket &=&  \dot{I}|\phi_k \ket + I |\dot{\phi}_k \ket \\
&=& -i [H,I]|\phi_k \ket +I|\dot{\phi}_k \ket \\
&=& -i \lambda_k H|\phi_k \ket + iIH|\phi_k \ket +I|\dot{\phi}_k \ket.
\en
Multiplying this from the left by $\bra \phi_p|$, we obtain
\bn 
 \lambda_k \bra \phi_p|\dot{\phi}_k \ket &=& -i \lambda_k \bra \phi_p|H|\phi_k \ket + i \lambda_p \bra \phi_p|H|\phi_k \ket \\
 & &+ \lambda_p\bra \phi_p|\dot{\phi}_k \ket \\
&=& (\lambda_p-\lambda_k) (\bra \phi_p|H|\phi_k \ket - i \bra \phi_p|\dot{\phi}_k \ket) \\
&=& 0.
\en
For $p \neq k$, we obtain
$$
\bra \phi_p|H|\phi_k \ket - i \bra \phi_p|\dot{\phi}_k \ket= 0.
$$
If the last equation held for $p = k$, we would have found
\bn
& &\sum_p |\phi_p \ket \bra \phi_p|H|\phi_k \ket =i \sum_p |\phi_p \ket \bra \phi_p|\dot{\phi}_k \ket \\
& & \to H |\phi_k \ket= i |\dot{\phi}_k \ket,
\en
which is not true in general. So let us try $ |\psi_k(t) \ket =
e^{i \alpha_k(t)}|\phi_k (t) \ket$ and require that $|\psi_k(t) \ket$
satisfies the Schr\"odinger equation;
$$
i\p_t |\psi_k \ket = -\dot{\alpha}_k e^{i \alpha_k}
|\phi_k \ket + e^{i \alpha_k} \partial_t|\phi_k \ket = H e^{i\alpha_k}
|\phi_k \ket.
$$
It then follows that
$
 \dot{ \alpha}_k(t) = \bra \phi_k(t)|\l( i \p_t - H \r)
|\phi_k(t) \ket$. Integrating this with respect to $t$, we obtain
$$ 
\alpha_k(t) = \int_0^t \bra \phi_k(s)|[i \partial_s - H(s) ]|\phi_k(s) \ket ds.
$$
\end{proof}
Let $|\psi (t)  \ket$ be an arbitrary solution of the Schr\"odinger equation. Since $\{|\phi_n(0) \ket\}$ is a complete set, $|\psi(0) \ket$
can be expanded as 
$$
|\psi(0)\ket = \sum_k c_k |\phi_k(0) \ket.
$$
From linearity, the solution at arbitrary $t >0$ is
\be
|\psi(t) \ket = \sum_k c_k e^{i\alpha_k(t)}|\phi_k(t) \ket,
\ee 
where $c_k$ is independent of time.

A few remarks are in order.
It is a common practice to write $|\psi(t) \ket$ as
$$
|\psi(t) \ket = U(t)|\psi(0)\ket, \ U(t) = {\mathcal T}e^{-i \int_0^t H(s) ds}.
$$
where $U(t)$ is the time-evolution operator. However, analytic evaluation of $U(t)$ is impossible, except for a few simple cases,
due to the time-ordering operation ${\mathcal T}$. In the above formalism, everything is found by solving the eigenvalue problem
$I|\phi_k(t) \ket = \lambda_k |\phi_k(t) \ket$ at a given instant
of time $t$, however, the difficultly lies in finding $I$.

Since $|\psi_k(0) \ket$ develops to $|\psi_k(t) \ket$ at time $t>0$,
the time-evolution operator $U(t)$ can be expressed as
\be
\sum_k |\psi_k(t) \ket \bra \psi_k(0) |
= \sum_k e^{i \alpha_k(t) }|\phi_k(t)\ket \bra \phi_k(0)|.
\ee

The evolution of $|\phi_k(t) \ket$ is transitionless since $U(t)|\phi_k(0)
\ket = e^{i \alpha_k(t)} |\phi_k(0)\ket$ for any $t>0$. The solution of
the Schr\"odinger equation always remains in the $k$-th eigenstate of
$I$ if $|\psi(0) \ket=|\phi_k(0) \ket$. On the other hand, $|\psi(t)\ket$
is not an eigenvector of $H$ and the time-evolution is nonadiabatic.

\subsection{Aharonov-Anandan Phase and Dynamical Invariants}

Let $U(T)$ be the time-evolution operator at a fixed time $T$ 
for a given Hamiltonian. $U(T)$ has the spectral decomposition
\be
U(t) = \sum_k e^{i \chi_k}|\chi_k \ket \bra \chi_k|,
\ee
where $U(T)|\chi_k \ket = e^{i \chi_k}|\chi_k\ket$ and $\bra \chi_k|
\chi_{k'} \ket = \delta_{kk'}$. Since
a unitary matrix is normal, the set of normalized eigenvectors
$\{|\chi_k \ket\}$ forms a complete set. 

Suppose the initial state of the wave function $|\psi(0)\ket$ is
$|\chi_k \ket$. Then we find
\bea
|\psi(T) \ket &=& U(T)|\psi (0) \ket = U(T)|\chi_k \ket \nonumber \\
&=& e^{i \chi_k} |\chi_k \ket= e^{i \chi_k}|\psi(0)\ket.
\eea 
By noting that $|\psi(T) \ket$ and $|\psi(0)\ket$ represent the
same vector in the projected Hilbert space $\C P^n =\C^n/{\rm U}(1)$, 
we find an eigenvector
of $U(T)$ executes a cyclic evolution in the projective Hilbert
space. Such a vector is called cyclic. 
In the U(1) fiber over $\{|\psi(0)\ket\}$,
$|\psi(0) \ket$ and $|\psi(T)\ket$ differ by a phase $e^{i \chi_k}$, 
which is called the Aharonov-Anandan phase \cite{aa}.

Let $|\phi(t) \ket$ be a closed
curve in the projective Hilbert space ($|\phi(T)\ket=|\phi(0) \ket$) 
such that 
$|\psi(t) \ket = e^{i \alpha(t) }|\phi(t) \ket$, where $|\psi(t)\ket$ is
a solution of the Schr\"odinger equation.
This $\alpha(t)$ is identified with the Lewis-Riesenfeld phase, meaning there is a dynamical
invariant $I$ whose eigenvector is $|\phi(t) \ket$. In fact
$U(T)$ can be written as
\bea
U(T) &=& \sum_k|\psi_k(T) \ket \bra \psi_k(0)| \nn \\
&=& \sum_k e^{i \alpha_k(T)}|\phi_k(T)\ket
\bra \phi_k(0)|\nn\\
&=&  \sum_k e^{i \chi_k}|\phi_k(0)\ket \bra \phi_k(0)|. 
\eea
Thus we have the following correspondences
\bn
|\phi_k(0) \ket &\leftrightarrow& |\chi_k \ket,\\
e^{i \alpha_k(T)} &\leftrightarrow& e^{i \chi_k}.
\en

Let
\be
\alpha_k(T)  = \int_0^T \bra \phi_k(t)|(i \partial_t - H)|\phi_k(t) \ket
dt
\ee
be the Lewis-Riesenfeld phase associated with the eigenvector $|\phi_k(t) \ket$ of $I$. 
The first term 
\be
\gamma_k^g = i \oint \bra \phi_k(t) | d|\phi_k(t) \ket
\ee
is reparameterization ($t \to \tau(t)$) invariant and geometric in nature
(geometric phase), while 
\be
\gamma_k^d = - \int_0^T \bra \phi_k(t)|H|\phi_k(t) \ket dt
\ee
is the dynamical phase.

When $\gamma_k^d = 0$ for all $k$, the time-evolution is called holonomic or geometric \cite{Ichikawa2012}.
A quantum gate satisfying this condition is called a holonomic (geometric) 
gate.

\section{Nonadiabatic Holonomic One-Qubit Gates}

For definiteness, let us consider \cite{nonad}
\be
H = \frac{1}{2} (\Omega \cos \omega t\ \sigma_x + \Omega \sin \omega t\ 
\sigma_y + \Delta\ \sigma_z).
\ee
It is easy to verify that 
\be
I =  \Omega \cos \omega t\ \sigma_x + \Omega \sin \omega t\ \sigma_y 
+ (\Delta-\omega) \ \sigma_z
\ee
is a dynamical invariant of $H$. The eigenvalues and eigenvectors of $I$ are
\be
\pm \lambda, \qquad |\phi_{\pm} (t) \ket = 
\begin{pmatrix}
e^{-i \omega t} \cos \theta_{\pm}\\
\sin \theta_{\pm} 
\end{pmatrix},
\ee
where $\lambda= \sqrt{\Omega^2+(\Delta-\omega)^2}$, 
$\cos \theta_{\pm} = \xi_{\pm}/\sqrt{1+\xi_{\pm}^2}$,
$\sin \theta_{\pm} = 1/\sqrt{1+\xi_{\pm}^2}$ with
$\xi_{\pm} = [(\Delta-\omega)\pm
\lambda]/\Omega$. The Lewis-Riesenfeld phases are readily evaluated as
\be
\alpha_{\pm} (t) = (\omega \mp \lambda )t/2.
\ee
Note that
$H, I$ and $|\phi_{\pm}(t) \ket$ are cyclic with a period 
$T=2\pi/\omega$.

Let us require that $U(T)$ is a holonomic gate, that is
\be
\gamma_\pm^d =  -\int_0^T \bra \phi_{\pm}(t)|H |\phi_{\pm}(t) \ket dt = 0.
\ee
This is satisfied if and only if 
\be
\Omega^2+\Delta (\Delta - \omega)=0. 
\ee
In fact, this condition not only satisfies $\gamma_\pm^d=0$
but also a stronger condition that the integrand
$\bra \phi_{\pm}(t)|H |\phi_{\pm}(t) \ket$ vanishes for 
$\forall t \in [0,T]$. 
We find from this condition that $\Delta \in [0,\omega]$, from which
we also find $\Omega^2 \sim \Delta \omega$. Adiabaticity cannot be attained
under this condition and hence such a gate cannot be realized within the adiabatic regime.

When the above conditions are met, the resulting gate is \cite{nonad}
\bea\label{eq:ub}
U_{\beta}(T) &=& \sum_{\pm} e^{i \alpha_{\pm} (T)}
 |\phi_{\pm}(0)\ket \bra \phi_{\pm}(0)| \nn\\
&=& - e^{i \pi \sin \beta [-\cos \beta \sigma_x + \sin \beta \sigma_z]},
\eea
where $\cos^2 \beta =\Delta/\omega, \beta \in [0, \pi/2]$.

$U_{\beta}(T)$ generates a 1-dimensional trajectory in SU(2) manifold as shown in
Fig.~1.

\begin{Figure}
\label{fig1}
\begin{center}
\includegraphics[width=8cm]{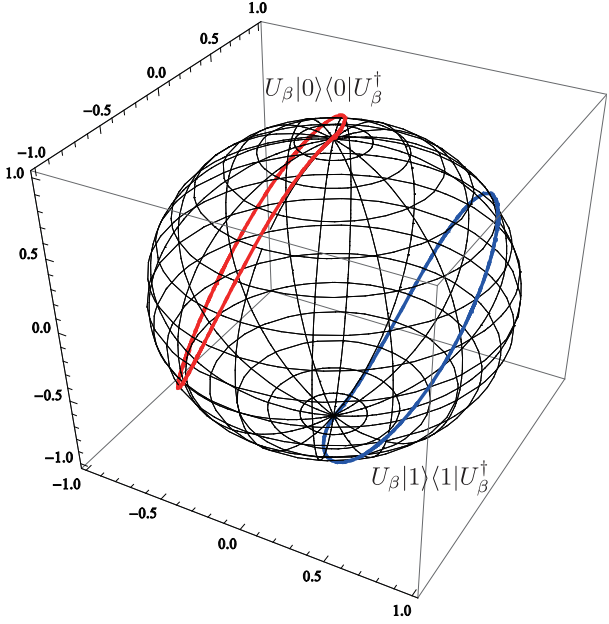}
\end{center}
\captionof{figure}{Trajectories of the Bloch vectors corresponding to
$U_{\beta}|0\ket \bra 0|U_{\beta}^{\dagger}$ and 
$U_{\beta}|1\ket \bra 1|U_{\beta}^{\dagger}$, $\beta \in
[0, \pi/2]$.}
\end{Figure}

By noting that
$[U_{\beta_1}(T), U_{\beta_2}(T)] \neq 0$ in general, the set $\{U_{\beta}(T)\}$ generates
all SU(2) group elements and hence forms a universal set of one-qubit gates.

\section{Examples}
In this section, we consider several important one-qubit gates. Since $U_{\beta}$ given
in Eq.~(\ref{eq:ub}) implements a one-dimensional subset of SU(2),
we need to employ several $U_{\beta_i}$ to implement arbitrary one-qubit
gates.

In what follows, we list $\{\beta_i\} = \{\beta_1, \ldots, \beta_N \}$ with the convention that $\beta_1$ acts first and $\beta_N$ acts last, and we give the gate fidelity $\mathcal F= \tr(U^\dagger U_\text{ideal})/
\tr(U_{\rm ideal}^\dagger U_{\rm ideal})$.
The numerical results below are high-fidelity implementations of the desired gates.

\subsection{NOT gate}
The NOT gate
\begin{align}
e^{i \pi /2}
\begin{pmatrix}
0&1 \\
1&0
\end{pmatrix}
\end{align}
can be realized by using 4 gates with $\{\beta_i\} =$ \{0.423, 0.680, 0.236, 0.222\}. The fidelity is $0.99999999990$.

\subsection{Hadamard gate}
A Hadamard gate 
\begin{align}
\frac{e^{i \pi /2}}{\sqrt{2}}
\begin{pmatrix}
1&1 \\
1&-1
\end{pmatrix}
\end{align}
with good fidelity $0.99999999791$ requires 7 gates, with $\{\beta_i\}$ given as \{0.331, 0.783, 0.300, 0.926, 0.174, 0.851, 0.347\}.

\subsection{Phase gate}
The phase gate
\begin{align}
e^{-i \pi /4}
\begin{pmatrix}
1&0 \\
0&i
\end{pmatrix}
\end{align}
implemented using 4 gates $\{\beta_i\} = $ \{0.827, 0.102, 0.287, 0.777\} has the fidelity $0.99999999993$.

\subsection{$\pi/8$-gate}
Similar to the NOT gate, $\pi/8$-gate
\begin{align}
e^{-i \pi /8}
\begin{pmatrix}
1&0 \\
0&e^{i \pi /4}
\end{pmatrix}
\end{align}
can be implemented using 3 elementary gates with $\{\beta_i\} = $ \{0.788, 0.514, 0.788\} and fidelity $0.99999999996$.

\section{Summary}

A review of holonomic gates implemented by using
the dynamical invariants is given. 
An interesting relation between the dynamical
phase and the Aharonov-Anandan phase is clarified.
We have explicitly shown that important one-qubit gates can be
implemented by combining holonomic quantum gates.

Analysis of the robustness of our one-qubit gates
under noise is an important issue
and will be published elsewhere.

\section*{Acknowledgement}

We are grateful to Yasushi Kondo for careful reading of
the manuscript.
UG and MN are grateful to the Japan Society for the
Promotion of Science (JSPS) for partial support from a
Grant-in-Aid for Scientific Research (Grant No. 24320008).
MN also thanks JSPS for a Grant-in-Aid for Scientific
Research (Grant No. 23540470). UG acknowledges the
financial support of the Ministry of Education, Culture,
Sports, Science and Technology (MEXT) Scholarship for
foreign students. 

\end{multicols}


\begin{thebibliography}{9}

\bibitem{lr}
H. R. Lewis and W. B. Riesenfeld (1969) J. Math. Phys. {\bf 10}, 
1458.


\bibitem{aa}
J. Anandan and Y. Aharonov (1990) Phys. Rev. Lett. {\bf65}, 1697.

\bibitem{Ichikawa2012} 
T. Ichikawa, M. Bando, Y. Kondo, and M. Nakahara (2012)
Phil. Trans. R. Soc. A {\bf 370},
4671.

\bibitem{nonad} U. G\"ung\"ord\"u, Y. Wan and M. Nakahara
 (2014) J. Phys. Soc. Jpn. {\bf 83}, 034001.



\end{thebibliography}
\end{document}